\pdfoutput=1
\RequirePackage{ifpdf}
\ifpdf 
\documentclass[pdftex]{sigma}
\else
\documentclass{sigma}
\fi

\begin{document}

\allowdisplaybreaks

\renewcommand{\thefootnote}{$\star$}

\newcommand{\arXivNumber}{1506.08675}

\renewcommand{\PaperNumber}{089}

\FirstPageHeading

\ShortArticleName{On the Relationship between Two Notions of Compatibility for Bi-Hamiltonian Systems}

\ArticleName{On the Relationship between Two Notions\\
 of Compatibility for Bi-Hamiltonian Systems\footnote{This paper is a~contribution to the Special Issue
on Analytical Mechanics and Dif\/ferential Geometry in honour of Sergio Benenti.
The full collection is available at \href{http://www.emis.de/journals/SIGMA/Benenti.html}{http://www.emis.de/journals/SIGMA/Benenti.html}}}

\Author{Manuele SANTOPRETE}

\AuthorNameForHeading{M.~Santoprete}

\Address{Department of Mathematics, Wilfrid Laurier University, Waterloo, ON, Canada}
\Email{\href{mailto:msantoprete@wlu.ca}{msantoprete@wlu.ca}}

\ArticleDates{Received June 30, 2015, in f\/inal form November 03, 2015; Published online November 07, 2015}

\Abstract{Bi-Hamiltonian structures are of great importance in the theory of integrable Hamiltonian systems. The notion of compatibility of symplectic structures is a key aspect of bi-Hamiltonian systems. Because of this,
   a few dif\/ferent notions of compatibility have been introduced. In this paper we show that, under some additional assumptions, compatibility in the sense of Magri implies a notion of compatibility due to Fass\`o and Ratiu, that we dub bi-af\/f\/ine compatibility. We present two proofs of this fact. The f\/irst one uses the uniqueness of the connection parallelizing all the Hamiltonian vector f\/ields tangent to the leaves of a~Lagrangian foliation. The second proof  uses  Darboux--Nijenhuis coordinates  and   symplectic connections.}


\Classification{70H06; 70G45; 37K10}

\renewcommand{\thefootnote}{\arabic{footnote}}
\setcounter{footnote}{0}

\section{Introduction}
Let $ M $ be a smooth manifold of dimension $ 2n $, and let $ X $ be a vector f\/ield on~$M$. Suppose that~$ X $ is Hamiltonian with respect to two dif\/ferent symplectic structures~$ \omega _1 $ and $~\omega _2$, that is,
\begin{gather*}
    \mathbf{i} _{ X } \omega _1 = \mathbf{d} H \qquad \text{and} \qquad \mathbf{i} _{ X } \omega _2 = \mathbf{d} K,
\end{gather*}
where $ H $ and $ K $ are two, possibly distinct, Hamiltonian functions. Let us  introduce the so called
{\it recursion operator}  $ N = \omega _2 ^\sharp \omega _1 ^\flat \colon TM \to TM $, where  $ \omega ^\flat \colon TM \to T ^\ast M $ denotes the ``musical'' isomorphism induced by the symplectic form $ \omega $, and $ \omega ^\sharp $ is its inverse.  Then, as a consequence of the fact that $ X $ is Hamiltonian with respect to two symplectic forms, the flow  associated to $ X $ preserves the eigenvalues of the recursion operator. Hence, if $ N $ has  $n $  functionally independent eigenvalues  in involution, then, it is completely integrable via the Liouville--Arnold theorem.

A natural approach to integrability is to try to f\/ind suf\/f\/icient conditions for the eigenvalues of the recursion operator to be in involution. Several  suf\/f\/icient conditions of this type have been found.

One such condition is  based on the pioneering work of Magri \cite{Magri_simple_1978} in the inf\/inite-dimensional case. Magri and Morosi \cite{Magri_geometrical_1984} showed that, if  the sum of the Poisson tensor associated to $ \omega _1 $ and the one associated to $ \omega _2 $ is still a Poisson tensor, then the eigenvalues of the recursion operator are in involution. A similar claim is also present in the work of Gel'fand and Dorfman \cite{Gelfand_Hamiltonian_1979}. In this case we say that $ \omega _1 $ and $ \omega _2 $ are {\it Magri-compatible},  the triple $ (M, \omega _1 , \omega _2) $ is a {\it bi-Hamiltonian manifold}, and the quadruple  $ (M , \omega _1 , \omega _2 , X) $  is a {\it bi-Hamiltonian system} in Magri's sense if there exist functions $ H _1 $ and $ H _2 $ such that $ X = \omega _1 ^\sharp  \cdot \mathbf{d} H _1 = \omega _2 ^\sharp \cdot   \mathbf{d} H _2$. However, it is known that not all completely integrable Hamiltonian system are bi-Hamiltonian in Magri's sense. In fact there are results by Brouzet~\cite{Brouzet_Systemes_1990} and Fernandes~\cite{Fernandes_Completely_1994} indicating that there are completely integrable Hamiltonian system that are not bi-Hamiltonian in Magri's sense.  This limitation of Magri's def\/inition stimulated the search for dif\/ferent notions of compatibility.

Bogoyavlenskij has given a suf\/f\/icient condition that he calls {\it strong dynamical compatibi\-li\-ty}~\cite{Bogoyavlenskij_Theory_1996}. This condition requires the existence of a vector f\/ield $X$ that is Hamiltonian with respect to two symplectic structures~$ \omega _1 $ and~$ \omega _2 $, is completely integrable with respect to $ \omega _1 $, and it is {\it non-degenerate} in the following sense. The orbits of~$X$ lie on Lagrangian tori and, in any local system of $ \omega _1 $-action-angle coordinates~$ (a, \alpha)$, the Hamiltonian function~$ H _1 $ of $X$  associated to~$ \omega _1 $ satisf\/ies the following  equation
\begin{gather*}
    \det \left( \frac{ \partial ^2 H _1 } { \partial a _i \partial a _j } \right)(a) \neq 0.
\end{gather*}

A third notion of compatibility, was introduced by Fass\`o and Ratiu (see \cite{Fasso_Compatibility_1998}) in order to study superintegrable systems, that is, systems with  more than $n$ independent integrals of motion, and with motions  on isotropic tori of dimension less than $ n $, rather than on Lagrangian tori of dimension $n$. Let $ \omega _1 $ and $ \omega _2 $ be two symplectic forms on $M$. We say that a f\/ibration (foliation) is bi-Lagrangian if the f\/ibers (leaves) are Lagrangian with respect to both symplectic forms. Suppose there exist a bi-Lagrangian f\/ibration of~$ M $. We say that $ \omega _1 $ and $ \omega _2 $ are {\it bi-affinely compatible}  if the Bott connection  (see Section~\ref{section2} for a~def\/inition)  associated to $ \omega _1 $ and the one associated to $ \omega _2 $ coincide.

Later we will explain the notion of Magri compatibility and bi-af\/f\/ine compatibility  in more detail.
In  \cite{Fasso_Compatibility_1998} Fass\`o and Ratiu wrote:

``It is not known to us whether our def\/inition (as well as Bogoyavlenskij's) is more general than Magri's. If $ \omega _1 $ and $ \omega _2 $ are Magri compatible, then the eigenvalues of the recursion ope\-ra\-tor (if independent) def\/ine a bi-Lagrangian foliation. However, even when this foliation is a~f\/ibration and has compact and connected f\/ibers, it is not clear whether, using the terminology of Def\/inition~2, it is bi-af\/f\/ine.''

This paper will be devoted to showing that Magri compatibility implies bi-af\/f\/ine compatibility in the simple case where the recursion operator has the maximal number of distinct eigenvalues. Note, however, that the converse is not true, since, as shown in~\cite{Bogoyavlenskij_Theory_1996} and \cite{Fasso_Compatibility_1998},  there exist bi-af\/f\/inely compatible structures that are not compatible in Magri's sense. We believe it should be possible to tackle the general case by using Turiel's classif\/ication of Magri-compatible bi-Hamiltonian structures (see \cite{Olver_Canonical_1990,Turiel_Classification_1994}).

We are aware of three dif\/ferent ways of proving that Magri compatibility implies bi-af\/f\/ine compatibility.
The f\/irst proof uses the property that the  Bott connection is the  unique   connection parallelizing all the Hamiltonian vector f\/ields tangent to the leaves of a Lagrangian foliation.
The second proof employs Darboux--Nijenhuis coordinates and the fact (proved in Proposition~\ref{prop:Bott}) that  the restriction of a torsion-free symplectic connection to an involutive Lagrangian distribution coincides with the Bott connection in~$L$.
A third proof can be obtained  directly by  using  the def\/inition of Bott's connection given in  equation~\eqref{eqn:Bott_connection_0}.
In this paper we will  present only the f\/irst two proofs.

The f\/irst proof is more direct and has the advantage of avoiding the introduction of Darboux--Nijenhuis coordinates and symplectic connections. The second proof, on the other hand, shows that  Magri's compatibility condition allows to construct explicitly, in Darboux--Nijenhuis coordinates,  two symplectic connections that  have a special form.

Recall that, while Magri compatibility implies bi-af\/f\/ine compatibility, the converse, as mentioned above, does not hold.  Since the restriction of a symplectic connection to an integrable Lagrangian distribution~$L$  is the Bott connection in $L$ (as shown in  Proposition~\ref{prop:Bott}) and  bi-af\/f\/ine compatibility, by def\/inition,  concerns itself only with the Bott connection, it seems clear that that the dif\/ference between the two types of compatibility lies in the restrictions Magri's condition imposes on the allowable symplectic connections. Therefore,  in our opinion, the link between symplectic connections and the notion of compatibility for bi-Hamiltonain systems deserves further investigation.

\section{Partial and symplectic connections}\label{section2}

In this section we briefly review some facts about partial connections and symplectic connections. We refer the reader to~\cite{Forger_Lagrangian_2013} for further details.

\subsection{Bott connection}

We f\/irst recall the def\/inition of a partial connection.
\begin{definition} Let $M$ be a manifold, $V$ a vector bundle over $M$, $L$ a distribution on~$M$. Let $ \Gamma (V) $  denote the space of sections of $V$, and $ \Gamma (L) $ the space of vector f\/ields tangent to~$L $. A~{\it partial linear connection } in $V$ along $L$ is a bilinear map
    \begin{gather*}
        \nabla \colon \  \Gamma (L)   \times   \Gamma (V) \to \Gamma (V)\colon \ (X, Y) \to \nabla _X Y,
    \end{gather*}
    such that
    \begin{enumerate}\itemsep=0pt
        \item[1)] $ \nabla _{ f X }  Y = f \nabla _{ X } Y $ (i.e., $ \nabla $ is linear in $X$),
        \item[2)]  $ \nabla _X (f Y) = f \nabla _X Y + (X [f]) Y $ (i.e., $ \nabla $  is a derivation),
    \end{enumerate}
    for $ X \in \Gamma (L) $, $ f \in C ^{ \infty } (M) $, and $ Y \in \Gamma (V) $.
\end{definition}

Suppose $ L $ is an involutive distribution on $M$, and let $ L ^\perp $ be the {\it annihilator} of~$L$, that is, the vector subbundle of~$ T ^\ast M $  consisting of $1$-forms that vanish on $L$. The partial linear connection~$ \tilde \nabla ^B $ in~$ L ^\perp $ along~$L$ def\/ined by
\begin{gather*}
    \tilde \nabla _X ^B \alpha = \mathcal{L} _{ X } \alpha, \qquad \mbox{for} \quad X \in \Gamma (L) , \quad \alpha \in \Gamma \big(L ^\perp\big)
\end{gather*}
is the {\it Bott connection} associated with the distribution $L$.

Now assume that $ ( M , \omega) $ is an almost-symplectic manifold (that is, $ \omega $ is a non-degenerate 2-form), and~$L$ is an involutive Lagrangian distribution with respect to~$\omega$. By Frobenius theorem $L$ is also integrable, that is, each point of~$M $ is contained in an integral manifold of $L $. Moreover, the collection of all maximal connected integral manifolds of $L$ forms a  foliation of~$M$ (see~\cite{Lee_Introduction_2003} for more details).  Since  the ``musical'' isomorphism $ \omega ^\flat \colon T M \to T ^\ast M $ and its inverse $ \omega ^\sharp \colon T ^\ast M \to T M $ restrict to $ \omega ^\flat \colon \Gamma (L)   \to \Gamma (L ^\perp) $ and  $ \omega ^\sharp \colon \Gamma (L ^\perp) \to \Gamma (L)  $ we can use them  to def\/ine a partial linear connection in~$L$ along $L$.

\begin{definition} \label{dfn:Bott_connection} Let $L$ be an involutive  Lagrangian distribution on the almost symplectic mani\-fold~$(M , \omega) $. The partial connection
    \begin{gather*}
        \nabla ^B \colon \ \Gamma (L) \times \Gamma (L) \to \Gamma (L) \colon \ (X,Y) \to \nabla _X Y
    \end{gather*}
    def\/ined by
    \begin{gather}\label{eqn:Bott_connection_0}
        \nabla _X ^B Y = \omega ^\sharp \mathcal{L} _{ X } (\omega ^\flat Y)
    \end{gather}
    is  called, by an abuse of terminology, the {\it Bott connection} in $L$.
\end{definition}

The Bott connection can also be def\/ined with the following formula:
\begin{gather}\label{eqn:Bott_connection}
    \omega \big(\nabla _X ^B Y , Z\big) = X [\omega (Y, Z) ] - \omega (Y , [X,Z])
\end{gather}
for $ X, Y \in \Gamma (L)$, $Z \in \Gamma (T M) $.

It can be shown that $ \nabla ^B $  def\/ines a flat partial connection. See~\cite{Forger_Lagrangian_2013} for more details.
However, in general, the Bott connection is not torsion-free, in fact we have the following proposition.

\begin{proposition}
  Let $(M, \omega) $ be an almost symplectic manifold, and let $L$ be an involutive Lagrangian distribution on~$M$. Then, if $ \omega $ is closed, the Bott connection $ \nabla ^B $ in $L$ has zero torsion. More generally, the torsion tensor $  T ^B$ of $ \nabla ^B $, defined by
  \begin{gather*}
  T ^B (X,Y)  = \nabla _X ^B Y - \nabla _Y ^B X - [X,Y]
   \end{gather*}
  is related to the exterior derivative of $ \omega $ by the formula
  \begin{gather*}
      \mathbf{d} \omega (X, Y , Z) = \omega \big(T ^B (X,Y) , Z\big)\qquad  \text{for} \quad X,Y \in \Gamma (L) , \quad Z \in \Gamma (T M).
  \end{gather*}
 \end{proposition}

A proof of this statement can be found in~\cite{Forger_Lagrangian_2013}.

Another important property of the Bott connection is that it is the unique connection parallelizing all the Hamiltonian vector f\/ields tangent to the leaves of a Lagrangian foliation. More precisely we have the following theorem.

\begin{theorem}\label{thm:Lagrangian_foliation}
   Suppose $ (M , \omega) $ is a $2n$-dimensional symplectic manifold, and $U$ an open subset of~$M$. Let $ f _1 , \ldots , f _n $ be a  set of smooth functions defined on $U$  such that
   \begin{enumerate}\itemsep=0pt
       \item[$1)$] the $ f _i $ are pairwise in involution, that is, $ \{ f _i , f _j \} = 0 $,
       \item[$2)$] the differentials are linearly independent, that is, $ \mathbf{d} f _i \wedge \cdots \wedge \mathbf{d} f _n \neq 0 $.
   \end{enumerate}
   Then the $ f _i $ define a Lagrangian foliation of $U $, with leaves of the form
   \begin{gather*}
       N ^c  = \{ x\,|\,f _1(x) = c _1 , \ldots, f _n(x) = c _n \},
   \end{gather*}
  and the  Bott connection
   \begin{gather*}
       \nabla ^B  _X Y = \omega ^\sharp \mathcal{L} _X (\omega ^\flat Y)
   \end{gather*}
   is the unique connection parallelizing all the Hamiltonian vector fields tangent to the leaves.
\end{theorem}

\begin{proof}
    The fact that the $ f _i $ def\/ine a Lagrangian foliation is well known.
    Let $ Y = \omega ^\sharp \cdot \mathbf{d} g $ be a Hamiltonian vector f\/ield tangent to the leaves of the foliation.
   We prove that the Bott connection parallelizes the  Hamiltonian vector f\/ields tangent to the leaves.
   Clearly, we have
   \begin{gather*}
       \nabla ^B  _{ X _i } Y   = \omega ^\sharp \mathcal{L} _{ X _i }  (\omega ^\flat  Y) = \omega ^\sharp \mathcal{L} _{ X _i }  (\mathbf{d} g )
       \overset{\text{(by Cartan's formula)}}{=} \omega ^\sharp \big(\mathbf{i} _{ X _i } \mathbf{d} ^2 g + \mathbf{d} (\mathbf{i} _{ X _i } \mathbf{d} g)\big)       \\
\hphantom{\nabla ^B  _{ X _i } Y}{}
 =  \omega ^\sharp ( \mathbf{d} (\mathbf{i} _{ X _i } \mathbf{d} g)) = \omega ^\sharp (\mathbf{d}  \langle \mathbf{d} g , X _i  \rangle)
         = \omega ^\sharp (\mathbf{d}  \langle \omega ^\flat Y , X _i  \rangle)
         = \omega ^\sharp (\mathbf{d} (\omega (Y , X _i))) = 0,
   \end{gather*}
   since $ \omega (Y , X _i) = 0 $. This is because  $Y$ is tangent to the leaves and the leaves are Lagrangian submanifolds.
   Now, suppose $ X = \sum\limits_{ i = 1 } ^n  a _i X _i $, then
   \begin{gather*}
   \nabla _X Y = \sum a _i \nabla ^B _{ X _i } Y = 0.
   \end{gather*}
  Hence, $ \nabla ^B $ parallelizes all the Hamiltonian vector f\/ields $ Y = \omega ^\sharp \cdot \mathbf{d} g $ tangent to the leaves.

  To show uniqueness suppose there is another connection $ (\nabla ^B )'  $ with the same property. Then, there is  a tensor $ S (X, Y) $ such that $  (\nabla ^B )' _X Y = \nabla ^B _X Y + S (X, Y) $.
Hence, we have
\begin{gather*}
    0= \big(\nabla ^B \big)' _{X _i }  X _j  = \nabla ^B _{X _i} X _j    + S (X _i , X _j) = S (X _i , X _j).
\end{gather*}
Consequently, $ S (X, Y) = 0 $ for all $ X , Y $. In fact, let $ X = \sum \limits_{ i = 1 } ^n a _n X _i $ and $ Y = \sum\limits _{ i = 1 } ^n b _i X _i $, then
\begin{gather*}
    S (X, Y) = \sum _{ i,j } a _i b _j S (X _i , X _j) = 0
\end{gather*}
by linearity, since $  S (X _i , X _j)=0$.
\end{proof}

\begin{remark}
   Suppose the hypotheses of the theorem above are verif\/ied and let $ Y _1, \ldots, Y _n$ be Hamiltonian vector f\/ields that, at each point, span the tangent plane to the leaves of the foliation. Then, the Bott connection is the unique connection parallelizing all the vector f\/ields~$ Y _i $. The proof of this is given by a computation similar to the one in the proof of the theorem above.
\end{remark}

\subsection{Symplectic connections}
\begin{definition}
    Let $ (M , \omega) $ be an almost symplectic manifold. A {\it symplectic connection} on $ (M , \omega) $ is a bilinear map
    \begin{gather*}
    \nabla \colon \  \Gamma (T M) \times \Gamma (TM) \to \Gamma (TM)\colon \  (X ,Y) \to \nabla _X Y
    \end{gather*}
such that $ \nabla $ is a linear connection, that is,
\begin{enumerate}\itemsep=0pt
    \item[1)] $ \nabla _{ f X } = f \nabla _X Y $,
    \item[2)] $ \nabla _X (f Y) = f \nabla _X Y + (X [f]) Y$,
\end{enumerate}
and parallelizes $ \omega $, that is,
\begin{enumerate}
    \item[3)] $ (\nabla _X \omega) (Y,Z)=0$ for all $ X,Y ,Z \in \Gamma (TM) $, where $ \nabla _X \omega $ denotes the covariant derivative of~$ \omega $, given by the formula
        \begin{gather*}
        (\nabla _X \omega) (Y,Z) = X [\omega (Y,Z)] - \omega (\nabla _X Y , Z) - \omega (Y , \nabla _X Z).
        \end{gather*}
\end{enumerate}
\end{definition}
Note that here we adhere to the terminology of \cite{Forger_Lagrangian_2013}. Other authors incorporate the requirement of being torsion-free,
namely
\begin{gather*}
    T (X, Y) = \nabla _X Y - \nabla _Y X - [X,Y] = 0
\end{gather*}
for all $ X,Y  \in \Gamma (TM) $,  in the def\/inition of a symplectic connection.

The existence of torsion-free symplectic connections is ensured, in the case of symplectic manifolds, by the following proposition.
\begin{proposition}
  Let $ (M , \omega) $ be a symplectic manifold. Then, there is a torsion-free symplectic connection on~$M$.
\end{proposition}

See \cite{Bieliavsky_Symplectic_2006} for a proof.
On the other hand, if the manifold is not symplectic (i.e., $ \omega $ is not closed) a torsion-free symplectic connection does not exist, in fact we have the following proposition.

\begin{proposition}
  Let $ (M , \omega) $ be an almost symplectic manifold, and let $ \nabla $ be a symplectic connection on $ M $. Then, the torsion tensor $T$ of $ \nabla $ is related to the exterior derivative of $ \omega $ by the formula
  \begin{gather*}
  \mathbf{d} \omega (X,Y ,Z) = \omega (T (X,Y) , Z) + \omega (T (Z ,Y) , X)+ \omega (T (Z, X) , Y).
  \end{gather*}
Consequently, if there is  a torsion-free symplectic connection $ \nabla $ on $M$, $ \omega $ must be closed.
\end{proposition}

Suppose $L$ is a distribution on $M$. Let $\Gamma  (L) $ denote the set of all vector f\/ields tangent to~$L$, and let $ \Gamma (TM)    $ denote the set of all vector f\/ields on~$M$.
Recall that a connection $ \nabla $ is said to {\it parallelize $($or preserve$)$ a distribution}~$L$ if  $ \nabla _X Y \in \Gamma (L)   $ for all vector f\/ields $ X \in \Gamma  (TM) $  and all vector f\/ields in~$Y \in \Gamma  ( L)$.
The following result, proved in~\cite{Forger_Lagrangian_2013}, links symplectic connections and Lagrangian distributions.

\begin{proposition}
  Let $(M , \omega) $ be an almost symplectic manifold and let $L$ be a Lagrangian distribution on~$M$. If there exists a torsion-free symplectic connection $ \nabla $ on $M$ that preserves~$L$, then~$\omega $ must be closed $($that is, symplectic$)$, and~$L$ must be involutive.
  \end{proposition}

Finally, we state and prove a proposition, also proved in~\cite{Forger_Lagrangian_2013}, that will be essential for the main result of this paper.

  \begin{proposition}\label{prop:Bott}
Let $ (M , \omega) $ be a symplectic manifold and let~$L$ be an involutive Lagrangian distribution, then the restriction of any symplectic torsion-free connection~$ \nabla $ preserving~$L$ to~$L$ coincides with the Bott connection in $L$.
\end{proposition}

\begin{proof}
   From the def\/inition of   symplectic connection:
   \begin{gather*}
       \omega (\nabla _X Y, Z)   = X [\omega (Y,Z) ] - \omega (Y, \nabla _X Z)
      =     X [\omega (Y,Z) ] - \omega (Y, \nabla _Z X) - \omega (Y , [X,Z]),
   \end{gather*}
   where the second equality holds because, by hypothesis, the connection is torsion-free. Since~$ \nabla $ preserves~$L$ we have that, if $ X \in \Gamma (L)  $, then $ \nabla _Z X \in \Gamma (L) $. Furthermore, if  $ X,Y \in \Gamma (L) $, then, since the distribution is Lagrangian, we have that  $  \omega (Y, \nabla _Z X) =0 $, and
   \begin{gather*}
        \omega (\nabla _X Y, Z) =  X [\omega (Y,Z) ] - \omega (Y , [X,Z]).
   \end{gather*}
   This last equation agrees with equation~\eqref{eqn:Bott_connection} for the Bott connection.
\end{proof}

\section[Darboux-Nijenhuis coordinates]{Darboux--Nijenhuis coordinates}\label{sec:Darboux--Nijenhuis_coordinates}

 Suppose $ \omega _1 $ and $ \omega _2 $ are Magri-compatible symplectic forms on a smooth $ 2n $-dimensional mani\-fold~$ M $, and let $ N = \omega _2 ^\sharp \omega _1 ^\flat \colon  TM \to TM $ be the recursion operator.

 Magri's notion of  compatibility  can be equivalently expressed by saying that the {\it  Nijenhuis torsion} of the recursion operator vanishes for all vector f\/ields $ X$, $Y $, that is,
  \begin{gather*}
     T _N (X , Y) = [NX, N Y ] - N [NX, Y ] - N [X, NY] + N ^2 [X,Y] = 0.
 \end{gather*}
A proof of this fact can be found, for instance, in \cite{Magri_Eight_2004}.
 A tensor f\/ield with vanishing torsion is called a {\it Nijenhuis} tensor f\/ield.
 A Nijenhuis tensor f\/ield is {\it compatible} with a symplectic form~$ \omega $~if
 \begin{enumerate}\itemsep=0pt
          \item[1)] $ \omega ^\flat  N = N ^\ast \omega ^\flat$,
          \item[2)] $ \mathbf{d} \omega (N X , Y, \cdot) - \mathbf{d} \omega (NY,X, \cdot) + \mathbf{d} \Omega (Y , X, \cdot ) = 0 $ for all $ X$, $Y $,
  \end{enumerate}
  where $ N ^\ast $ is the adjoint tensor of $N$, and $ \Omega $ is def\/ined by the following expression
  \begin{gather*} \Omega (X , Y)   =  \langle \Omega ^\flat X, Y  \rangle =  \langle (\omega ^\flat N) X,Y \rangle.
  \end{gather*}
 The f\/irst condition ensures $ \Omega $ is skew symmetric, while the second ensures that $ \Omega $ is closed.

A triple $ (M , \omega , N) $, where $ \omega $ is a symplectic form and $ N $ is a compatible Nijenhuis tensor f\/ield, is called an   $\omega  N$ {\it manifold}. The def\/inition of $ \omega N $ manifold f\/irst appeared in the work of Magri and Morosi~\cite{Magri_geometrical_1984}.
A manifold $M$ with two Magri-compatible symplectic forms is an important example of an~$ \omega N $ manifold.

If $ (M, \omega _1 , \omega _2 , X) $ is  a bi-Hamiltonian system  (in Magri's sense) then it is possible to use the recursion operator to create a sequence of functions in involution that commute with  the Hamiltonians~$ H _1 $ and~$ H _2 $. Then, if a suf\/f\/icient number of integrals are functionally independent, the system is completely integrable. Such sequence of functions can be constructed by using the trace of powers of the recursion operator as follows
\begin{gather*}
I _k = \frac{ 1 } { k } \operatorname{Tr} \big(N ^k\big).
\end{gather*}
Since $ T _N (X, Y) = 0 $ it can be shown that the dif\/ferentials $ \mathbf{d} I _k $ satisfy the Lenard recursion relation $ N ^\ast  \mathbf{d} I _k = \mathbf{d} I _{ k + 1 } $, or equivalently, since $ N ^\ast = \omega _1 ^\flat \omega _2 ^\sharp $, they satisfy $ \omega _2 ^\sharp \cdot \mathbf{d} I _k = \omega _1 ^\sharp \cdot \mathbf{d} I _{ k + 1 } $.
Since the $ \mathbf{d} I _k $'s fulf\/ill the recursion relation it is easy to show that $ \{ I _i , I _j \} _1 = \{ I _i , I _j \} _2 = 0 $, where $ \{ \,,\, \} _1 $ and $ \{\, ,\, \} _2 $ are the Poisson brackets associated with $ \omega _1 $ and $ \omega _2 $, respectively. Thus, the~$ I _j $ are in involution. For a more detailed account of this construction we refer the reader to~\cite{Magri_Eight_2004, Magri_geometrical_1984}.
It is convenient to give the following def\/inition (see for example~\cite{Tondo_Generalized_2006}):
\begin{definition}
    A point $p$ of an $ \omega N $ manifold is a {\it regular point} if $N$ has the maximal number $ n = \frac{1}{2} \dim (M) $ of distinct, functionally independent, eigenvalues $ \lambda _1 , \lambda _2 ,\ldots , \lambda _n $ (i.e., the dif\/ferentials $ \mathbf{d} \lambda _1 , \ldots , \mathbf{d} \lambda _n $ are linearly independent at $p$). An open set $ U \subset M $ is called {\it regular} if each point of $U$ is a regular point.
\end{definition}

We remark that, in the def\/inition above, it is enough to require that the eigenvalues are distinct and non-constant, since the latter automatically implies their independence.

Furthermore, note that, in a regular open set $U$, since the~$ \lambda _i $ are functionally  independent, it follows  that the~$ I _i $'s are also  functionally independent (see, for instance,~\cite{Falqui_Poisson_2011} for a proof), and so the~$ I _i $ def\/ine a  bi-Lagrangian  foliation, namely, a foliation Lagrangian with respect to~$ \omega _1 $ and~$ \omega _2 $. In this setting, it can also be shown  that the eigenvalues are in involution $ \{ \lambda _i , \lambda _j \}  _1 = \{ \lambda _i , \lambda _j \} _2 = 0 $  (see~\cite{Falqui_Poisson_2011, Magri_Eight_2004}), and the $ \lambda _i $  also def\/ine the same  bi-Lagrangian  foliation.
With these def\/initions we have the following proposition.

\begin{proposition}\label{prop_DN}
  Let $ (M , \omega _1  , N) $ be an $ \omega N $ manifold. Each regular point has an open neighborhood where there exist coordinates $ (\boldsymbol{ \lambda } , \boldsymbol{ \mu }) = (\lambda _1 , \ldots , \lambda _n , \mu _1 , \ldots, \mu _n)$ $($where the $ \lambda _i $'s are the eigenvalues of $ N )$, called Darboux--Nijenhuis coordinates, satisfying the following two properties:
  \begin{enumerate}\itemsep=0pt
      \item[$1)$] $ \omega _1  = \sum _i \mathbf{d} \lambda _i  \wedge \mathbf{d} \mu _i  $, that is, they are Darboux-coordinates for $ \omega _1  $,
      \item[$2)$] $ N ^\ast \mathbf{d} \lambda _i = \lambda _i \mathbf{d} \lambda _i $ and $ N ^\ast \mathbf{d }  \mu _i = \lambda _i \mathbf{d} \mu _i $.
  \end{enumerate}
\end{proposition}

See \cite{Magri_Eight_2004} for a sketch of the proof of this statement.

In these coordinates the tensor $N$ takes the diagonal form
\begin{gather*}
    N = \begin{bmatrix}
        \boldsymbol{\Lambda}_n & {\bf 0}_n  \\
        {\bf 0}_n & \boldsymbol{\Lambda}_n
    \end{bmatrix},
\end{gather*}
where $ \boldsymbol{\Lambda}_n = \operatorname{diag } (\lambda _1 , \ldots , \lambda _n) $, and $ {\bf 0}_n$ is the $ n \times n $ matrix with zero entries. Moreover $ \omega _2 ^\sharp  = N (\omega _1  ^\flat) ^{ - 1 } $ takes the form
\begin{gather*}
    \omega _{ 2 } ^\sharp = \begin{bmatrix}
        {\bf 0}_n & \boldsymbol{\Lambda} _n \\
        - \boldsymbol{\Lambda}_n & {\bf 0}_n
    \end{bmatrix},
\end{gather*}
and since the matrix of $ \omega_2 ^\sharp  $ is the inverse of that of $ \omega_2 ^\flat $ and  the matrix of $   \omega_2 ^\flat $ is the negative of that of $ \omega_2 $ we have
\begin{gather*}
    \omega _{ 2 }  = \begin{bmatrix}
        {\bf 0}_n & \boldsymbol{\Lambda}_n ^{ - 1 }  \\
        - \boldsymbol{ \Lambda}_n^{ - 1 }   & {\bf 0}_n
    \end{bmatrix}.
\end{gather*}
\begin{remark}
The def\/inition of Darboux--Nijenhuis coordinates can be generalized to each set of Darboux coordinates in which $ N $ takes the diagonal form. With such more general def\/inition, one can manage the cases in which the eigenvalues are not independent.
\end{remark}

\section[Bi-affine compatibility]{Bi-af\/f\/ine compatibility}\label{section4}

We now recall the def\/inition of compatibility due to Fass\`o and Ratiu (see~\cite{Fasso_Compatibility_1998}) in the special case the f\/ibration is not only bi-isotropic but also bi-Lagrangian, since this is the only relevant case for our purposes. We refer to this type of compatibility as bi-af\/f\/ine compatibility.

\begin{definition}
    Let $M $ be a smooth manifold of dimension $ 2n $ and let $ \omega _1 $ and $ \omega _2 $ be two symplectic structures on $M $. Assume there exists a bi-Lagrangian f\/ibration $ \pi $ of $M$ with compact connected f\/ibers which is bi-af\/f\/ine (i.e., the restriction of  the Bott connections $ \nabla _{ 1 } ^B $ and $ \nabla _2 ^B $ associated with~$ \omega _1 $ and~$ \omega _2 $ to the f\/ibers, coincide).
    Then we say that $ \omega _1 $ and $ \omega _2 $ are {\it bi-affinely compatible } or {\it $\pi$-compatible}.
\end{definition}

\section[Relationship between Magri's notion of compatibility and bi-affine compatibility]{Relationship between Magri's notion of compatibility\\ and bi-af\/f\/ine compatibility}\label{section5}

Suppose we have a manifold $M$ with two Magri-compatible symplectic forms $ \omega _1 $ and~$ \omega _2 $. If every point is regular, then, as we mentioned in Section~\ref{sec:Darboux--Nijenhuis_coordinates}, the  $ I _k $'s (with $ I _k = \frac{ 1 } { k } \operatorname{Tr} (N ^k)$)  def\/ine a~bi-Lagrangian  foliation with associated distribution $L$. The link between Magri's notion of compatibility and bi-af\/f\/ine compatibility is given in the following theorem.

\begin{theorem}\label{thm:main0}
   Under the hypothesis above   the Bott connections
   \begin{gather*}
   \big(\nabla ^1 _B \big)_X Y = \omega _1 ^\sharp \mathcal{L} _{ X } (\omega _1  ^\flat Y)
   \qquad \mbox{and} \qquad \big(\nabla ^2 _B\big)_X Y = \omega _2 ^\sharp \mathcal{L} _{ X } (\omega _2  ^\flat Y)
   \end{gather*}
   in $L$ coincide. If, in particular, the Lagrangian foliation is a Lagrangian  fibration with compact connected fibers then Magri's compatibility implies bi-affine compatibility.
\end{theorem}

\begin{proof}
   Let $ X _k = \omega _1  ^\sharp \cdot \mathbf{d} I _{k + 1 } $ for $ k = 1 , \ldots, n $. Since  $ \omega _1 $ and $ \omega _2 $ are Magri compatible,  as we mentioned in Section~\ref{sec:Darboux--Nijenhuis_coordinates}, it follows that
   \begin{gather*} X _{ k  } = \omega _1 ^\sharp \cdot \mathbf{d} I _{ k +1 }= \omega _2 ^\sharp \cdot \mathbf{d}  I _{k}, \qquad k = 1, \ldots, n.  \end{gather*}
    Hence, the $ X _k $'s are Hamiltonian with respect to both symplectic structures and, by Theorem~\ref{thm:Lagrangian_foliation},  these vector f\/ields are parallel with respect to both~$\nabla ^1 _B $ and~$ \nabla ^2 _B $. Since the ~$ X _k $'s span the tangent space to the leaf, by the uniqueness results of  Theorem~\ref{thm:Lagrangian_foliation} and the following remark, it may be concluded  that the connections~$ \nabla ^1 _B $ and~$ \nabla ^2 _B $ coincide.
\end{proof}

Under the hypothesis of the theorem above, as a consequence of   Proposition~\ref{prop_DN}, we have that, in a neighborhood of  every point, there exist Darboux--Nijenhuis coordinates. As a consequence, we can give an alternative proof of the theorem above  that employs symplectic connections and  Darboux--Nijenhuis coordinates. This approach requires constructing explicitly the symplectic connections associated with~$ \omega _1 $ and~$ \omega _2 $. This will be done in the proof of the following theorem.

\begin{theorem}\label{thm:main}
   Under the hypothesis above   there are two torsion-free symplectic connections~$ \nabla ^1 $ and~$ \nabla ^2 $, symplectic with respect to~$ \omega _1 $ and~$ \omega _2 $, respectively, and  such that
   \begin{enumerate}\itemsep=0pt
       \item[$1)$] they preserve the  Lagrangian distribution $ L $  associated to the foliation defined by the eigenvalues $ \lambda  _1 , \ldots , \lambda _n $,
       \item[$2)$] the restrictions of~$ \nabla ^1 $ and~$ \nabla ^2 $ to  $L$ coincide with the Bott connections defined by $(\nabla ^1 _B )_X Y = \omega _1 ^\sharp \mathcal{L} _{ X } (\omega _1  ^\flat Y)  $ and  $(\nabla ^2 _B )_X Y = \omega _2 ^\sharp \mathcal{L} _{ X } (\omega _2  ^\flat Y)  $, respectively,
       \item[$3)$] the restrictions of~$ \nabla ^1 $ and~$ \nabla ^2 $ to~$L $ coincide, and thus $ \nabla ^1 _B = \nabla ^2 _B $.
   \end{enumerate}
\end{theorem}

In this approach the link between Magri's notion of compatibility and bi-af\/f\/ine compatibility  follows immediately from Theorem~\ref{thm:main} and is given in the following corollary.

\begin{corollary} If the bi-Lagrangian foliation is a fibration with compact connected fibers, then  Magri's compatibility implies bi-affine compatibility.
\end{corollary}

\begin{proof}[Proof of Theorem \ref{thm:main}] {\sloppy (1) We construct the symplectic connections explicitly. The fact that~$ \nabla ^1 $ and~$ \nabla ^2 $ are symplectic and preserve~$ L $ will follow from the construction. Take an atlas of~$M$ composed by Darboux--Nijenhuis charts. On each chart one can construct a torsion-free flat  connection, symplectic with respect to~$ \omega _1 $ and preserving~$L$, by taking the linear connection whose Christof\/fel symbols vanish identically in these coordinates. One can then obtain a global connection by using partitions of unity to ``glue'' the connections obtained in each Darboux--Nijenhuis chart. This construction gives the connection~$ \nabla ^1 $.  For a more detailed explanation of this  process see the proof of  Theorem~2 in~\cite{Forger_Lagrangian_2013}.

}

    We now explain the construction of~$\nabla ^2$ in detail. Suppose we are in a  Darboux--Nijenhuis chart. Let $ \mathbf{z} = (\boldsymbol{\lambda },\boldsymbol{\mu }) $, then  $ \mathbf{e} _i = \frac{ \partial } {\partial z ^i}$ is a basis of tangent vectors. In these coordinates, the vanishing of the covariant derivative of $ \omega _2  $ is
    \begin{gather} \label{eqn:cov_dev}
        \big(\nabla ^2  _{ \mathbf{e} _k } \omega _2 \big) _{ ij } = \frac{ \partial } { \partial z ^k  } (\omega_2) _{ ij } -  \sum _l \Gamma _{ i k } ^l (\omega_2) _{ l j } -  \sum _l \Gamma _{ j k } ^l (\omega_2) _{ il } =0,
    \end{gather}
   where the coef\/f\/icients $ \Gamma _{ ij } ^k $ are called Christof\/fel symbols. Since in Darboux--Nijenhuis coordinates we have
  \begin{gather*}
      \omega _{ 2 }  = \begin{bmatrix}
          {\bf 0}_n & \boldsymbol{ \Lambda}_n ^{ - 1 }  \\
          - \boldsymbol{ \Lambda}_n^{ - 1 }   & {\bf 0}_n
      \end{bmatrix},
  \end{gather*}
the term $   \frac{ \partial } { \partial z ^k  } (\omega_2 ) _{ ij } $ is always zero except for
   \begin{itemize}\itemsep=0pt
       \item $ k = i$, $j = i + n $ with $ 1 \leq i \leq n $, in which case  $  \frac{\partial } { \partial \lambda i } (\omega _2) _{ i (i + n ) } = - \frac{ 1 } { \lambda _i ^2 } $,
     \item  $ k = j$, $i = j + n $ with $ 1 \leq j \leq n $, in which case  $  \frac{\partial } { \partial \lambda j } (\omega _2) _{ (j + n)   j } = \frac{ 1 } { \lambda _j ^2 } $.
   \end{itemize}
   Let us take  all the $ \Gamma _{ ij } ^k $ to be zero except the ones of the form $ \Gamma _{ i i } ^i = -\frac{ 1 } { \lambda _i  } $ for $ 1 \leq i \leq n $. With this choice the Christof\/fel symbols are symmetric, and thus the connection is torsion-free.

   To verify that $ \nabla ^2  $ is symplectic with respect to~$ \omega _2 $ we  need to check only equation~\eqref{eqn:cov_dev} for a~few values of $ i$, $j $ and~$k$ , since for all the other values the equation is trivially verif\/ied.  We have
   \begin{gather*}
       \frac{ \partial} { \partial z ^i } (\omega_2)  _{ i (i + n) } -\Gamma _{ i i } ^i (\omega _2) _{ i(i + n)   }    =  \frac{ \partial} { \partial \lambda  _i } (\omega_2)  _{ i (i + n) } +\frac{ 1 } { \lambda _i ^2 } = -\frac{ 1 } { \lambda _i ^2 } +\frac{ 1 } { \lambda _i ^2 } =0,
   \end{gather*}
   where $ 1 \leq i \leq n $, and
    \begin{gather*}
       \frac{ \partial} { \partial z ^j} (\omega_2)  _{(j + n)   j } -\Gamma _{ j j } ^j (\omega _2) _{ (j + n)  j }    =  \frac{ \partial} { \partial \lambda  _j } (\omega_2)  _{ (j + n) j   } - \frac{ 1 } { \lambda _j ^2 } = \frac{ 1 } { \lambda _j ^2 } -\frac{ 1 } { \lambda _j ^2 } =0,
   \end{gather*}
   where $ 1 \leq j \leq n $.
   Therefore, the connection $ \nabla ^2  $  is symplectic with respect to  $ \omega _2 $.

    To verify that $ \nabla ^2 $ parallelizes $ L $ we use the following coordinate formula
    \begin{gather*}
        \nabla ^2 _{ \mathbf{e} _j } \mathbf{u} = \left( \frac{ \partial u ^i } { \partial  z ^j } + u ^k \Gamma ^i  _{ jk } \right) \mathbf{e} _{ i },
    \end{gather*}
    where, in general, $ \mathbf{u} =  \sum _i u ^i \mathbf{e} _i $ is a vector f\/ield in $\Gamma (TM) $. Now  suppose $ \mathbf{u} \in \Gamma  (L) $, then $ \mathbf{u} = \sum \limits{ i = n + 1 } ^{ 2n } u ^i \mathbf{e} _i $. Let $ \mathbf{v} =   \nabla ^2 _{ \mathbf{e} _j } \mathbf{u} $. Since $ \Gamma _{ i j } ^{ k } \neq 0 $  if and only if $ i = j = k $, it follows that $ \mathbf{v} $ is of the form
    \begin{gather*}
        \mathbf{v} = \sum _{ i = n + 1} ^{ 2n } a ^i \mathbf{e} _i.
    \end{gather*}
    Since $u ^k \neq 0 $ only if $ k \geq n + 1 $ and $ \Gamma _{ i j } ^k \neq 0 $ only when $ 1 \leq i \leq n $, then
    \begin{gather*}a ^i =
        \begin{cases}
             \dfrac{\partial  u ^i } { \partial z ^j },  & \text{if} \ \ i \geq n + 1,\vspace{1mm}\\
             0, & \text{in all other cases}.
        \end{cases}
    \end{gather*}
    Thus, $ \mathbf{v} \in \Gamma  (L) $, that is, $ \omega _2 $ parallelizes $L$.

  Now, taking an atlas covered with  Darboux--Nijenhuis charts, passing to a locally f\/inite ref\/inement $ (U _\alpha) _{ \alpha \in A } $,  denoting the corresponding family of linear connections constructed as above by $ ( \nabla  ^2 _\alpha) _{ \alpha \in A }$ and choosing a partition of unity $ (\chi _\alpha) _{ \alpha \in A } $ subordinate to the open covering $ (U _\alpha) _{ \alpha \in A } $, we can def\/ine
      \begin{gather*}
          \nabla ^2 = \sum _{ \alpha \in A } \chi _\alpha \nabla ^2  _\alpha.
      \end{gather*}
      The conditions of parallelizing a given dif\/ferential form, of parallelizing a given vector subbundle, and of being torsion-free are all local as well as af\/f\/ine. Thus, since each of the $ \nabla^2  _\alpha $ parallelizes~$ \omega _2 $ and~$L$, it follows that~$ \nabla ^2 $ also paral\-le\-li\-zes~$ \omega _2 $ and~$ L $.

    (2) This part of the theorem follows immediately from Proposition~\ref{prop:Bott}.

    (3) On each Darboux--Nijenhuis chart the restrictions of the connections~$ \nabla ^1 $ and~$\nabla ^2$ to the leaves of the foliation coincide, since for both connections the restriction of the Christof\/fel symbols  will be of the form   $ \widetilde \Gamma _{ ij } ^k = \Gamma _{ (i + n) (j + n)} ^{ k + n } $  with $ 1 \leq i, j , k \leq n $, and hence   the Christof\/fel symbols  $ \widetilde \Gamma _{ ij } ^k $ of both connections vanish identically in these coordinates.
\end{proof}

\subsection*{Acknowledgments}
We would like to thank one of the anonymous reviewers for suggesting to us that  Theorem~\ref{thm:main0} can be proved by using the uniqueness of the connection parallelizing all the Hamiltonian vector f\/ields tangent to the leaves of a Lagrangian foliation.
This work was supported by an NSERC Discovery Grant.

\pdfbookmark[1]{References}{ref}
\LastPageEnding

\end{document}